\documentclass{mathscan}

\usepackage{mathrsfs}
\usepackage{mathtools}
\usepackage{bbm}

\usepackage{theoremref}

\usepackage{enumitem}
\setlist[enumerate]{label= \alph*)}

\newtheorem{thm}{Theorem}[section]
\newtheorem{lemma}[thm]{Lemma}
\newtheorem{cor}[thm]{Corollary}
\newtheorem{prop}[thm]{Proposition}
\newtheorem{condition}[thm]{Condition}

\theoremstyle{definition}

\theoremstyle{remark}

\usepackage{etoolbox}
\usepackage{xifthen}

\DeclarePairedDelimiterX{\pair}[2]{\langle}{\rangle}{#1,#2}

\DeclareMathOperator*{\wslim}{w^\star-lim}

\newcommand{\supp}{\operatorname{supp}}

\newcommand{\overbar}[1]{\mkern 1.5mu\overline{\mkern-1.5mu#1\mkern-1.5mu}\mkern 1.5mu}

\newcommand{\R}{\mathbb{R}}

\newcommand{\N}{\mathbb{N}}

\newcommand{\B}{\mathcal{B}}

\newcommand{\sch}[1]{\langle #1 \rangle}
\newcommand{\eps}{\varepsilon}

\renewcommand{\Im}{\operatorname{Im}}
\renewcommand{\Re}{\operatorname{Re}}

\date{\today}

\AtBeginDocument{%
	\def\MR#1{}
}

\numberwithin{equation}{section}

\begin{document}
\begin{abstract}
	We show Rellich's theorem, the limiting absorption principle, and a Sommerfeld uniqueness result for a wide class of one-body Schrödinger operators with long-range potentials, extending and refining previously known results. Our general method is based on elementary commutator estimates, largely following the scheme developed recently by Ito and Skibsted.
\end{abstract}

\title[LAP and radiation conditions]{Limiting absorption principle and radiation conditions for Schrödinger operators with long-range potentials}

\author{Martin Dam Larsen}
\address{Department of Mathematics\\
	University of Copenhagen\\
	Universitetsparken 5\\
	DK-2100 Copenhagen Ø\\
	Denmark}
\email{mdl@math.ku.dk}

\maketitle
\tableofcontents

\section{Introduction}
We shall study the one-body Schrödinger operator
\[
H = H_0 + V, \quad H_0 = p^*p = -\Delta, \quad p = -i \nabla
\]
on $L^2(\R^d)$, $d\geq 1$, for a real bounded potential $V$ under natural minimal decay assumptions. We employ the notation
\[
r = r(x) = \sch{x} = \big(1+\abs{x}^2\big)^{1/2}, \quad \omega = \nabla r = \frac{x}{r},
\]
and use the convention that big- and small-O notation include a local boundedness assumption. We consider the following decay conditions on the potential $V$.
\begin{condition}\thlabel{cond1}
There is a radial function $W_0 = W_0(r)\geq 0$ and a splitting $V = V^{sr} + V^{lr}$ into real bounded functions such that:
\begin{enumerate}
	\item $W_0\in L^1([1,\infty),dr) $ and $W_0 = o(r^{-1})$ as $\abs{x} \to \infty$.
	\item $V^{lr}$ has first order distributional derivatives in $L^1_{loc}$, and 
	\[
	V^{lr} = o(1) = o(r^0) \text{ as } \abs{x}\to \infty, \quad \omega \cdot \nabla V^{lr} \leq W_0.
	\]
	\item $V^{sr}$ satisfies the short-range condition
	\[
	\abs{V^{sr}}\leq W_0.
	\]
\end{enumerate}
\end{condition}

\begin{condition}\thlabel{cond2}
In addition to \thref{cond1} the long-range part satisfies $\abs{\nabla V^{lr}} \leq W_0$.
\end{condition}
The requirement $V = o(1)$ ensures that $V$ is relatively $H_0$-compact, whence $H = H_0 + V$ defines a self-adjoint operator with domain $D(H) = D(H_0) = H^2(\R^d)$ and essential spectrum $\sigma_{ess}(H) = [0,\infty)$. We shall study the strictly positive spectrum of $H$. To this end, \thref{cond1,cond2} are very natural; the $o(r^{-1})$ condition exactly excludes positive point spectrum, and an $L^1$ assumption naturally appears in the one-dimensional (or radial) scattering problem.

The main goal of the present paper is to establish the limiting absorption principle in the optimal Besov setting and uniquely characterize the limiting resolvents as the solutions to the inhomogeneous Helmholtz equation that satisfy incoming/outgoing radiation conditions. To archive this, we establish a series of 'sharp' spectral results, namely Rellich's theorem (cf.\ \cite{Rellich43,Kato59,Agmon70}) on the absence of a certain class of slowly decaying eigenfunctions, and a number of 'smoothness bounds' (LAP bounds) on \emph{individual} resolvents in the optimal Besov setting. These results are interesting in their own right, and, together with the LAP, serves as key ingredients in the development of stationary scattering theory. We remark that all of our theorems are well known under the classical and extensively studied long-range condition $W_0 = O(r^{-1-\eps})$. See for instance \cite{Lavine73,AgmonHormander,Hormander4,Isozaki80,Saito79,Tao} and the references therein. As such, the achievement of the present paper is to extend a series of 'sharp' results, which usually require $O(r^{-1-\eps})$ decay in the literature, to decay dictated by \thref{cond1,cond2}. In particular, we cover the interesting case $W_0 = O(r^{-1}(\log r)^{-1-\eps})$. The formulations of our theorems should be compared to that of \cite{ItoSkibsted20} which served as an inspiration. 

Decay dictated by \thref{cond2} have previously been studied by Vakulenko \cite{Vakulenko} in the short-range case and Yafaev \cite[Chapter 11]{Yafaev_analytic} in the full long-range case. By exhibiting a carefully constructed \emph{strictly} positive commutator, they were able to establish the absence of positive eigenvalues as well as the following smoothness bound: For any compact set $I\subseteq (0,\infty)$, 
\[
\sup_{\lambda \in I, \, \eps >0} \norm{W_0^{1/2}\delta_\eps(H-\lambda)W_0^{1/2}} < \infty,
\]
where 
\[
\delta_\eps(H-\lambda) = \tfrac{1}{2\pi i} \big(R(\lambda+i\eps) - R(\lambda - i \eps)\big),\quad R(z) = (H-z)^{-1}.
\]
As remarked by Yafaev, their commutator method apparently only allows two-sided uniform estimates on $\delta_\eps$ (as opposed to individual resolvents). We remedy this situation in \thref{unif_bound,W_resolvent} stated below, while also extending the results to include the full LAP with radiation conditions. Finally, our version of Rellich's theorem (\thref{Rellich}) is a sharpening of the corresponding result on positive eigenvalues from \cite{Yafaev_analytic}. 

All of our main results follow from elementary commutator estimates requiring little more than integration by parts, largely based on the general scheme developed in \cite{ItoSkibsted20} by Ito and Skibsted. They study the limiting absorption principle and radiation conditions under 'classical' $O(r^{-1-\eps})$ long- and short-range conditions, although in a highly geometric setting. We trim away all unnecessary geometry in our flat Euclidean setting, which allows a simplified exposition. Secondly, we tweak the 'commutators with weight inside' in a subtle yet very powerful way to avoid higher order derivatives and archive stronger commutator lower bounds. Finally, the radiation condition estimates in \cite{ItoSkibsted20} (and elsewhere in the literature) heavily exploit the 'wiggle room' in the $O(r^{-1-2\eps})$ condition, that is, multiplicative factors of $r^{\eps'}$ for $\eps'<\eps$ can be handled without difficulty. We shall rediscover this wiggle room in \thref{cond2}. 

Before moving on we highlight some other instances of $L^1$-type decay conditions in the literature. Firstly, Agmon \cite{Agmon70}, extending the results of Kato \cite{Kato59} to the long-range case, established a version of Rellich's theorem under $L^1(dr)\cap o(r^{-1})$ decay. Agmon's result, considering $C^2$ solutions and requiring additional regularity on the potential $V$, is only slightly weaker than that of \thref{Rellich} below. Secondly, in \cite{Hormander2} on short-range scattering, Hörmander considers potentials $V = V^{sr}$ which map the unit open ball of
\[
\B^*_{H_0} = \{\psi \in H^2_{loc} \mid \norm{\psi}_{\B^*_{H_0}}<\infty\}, \quad \norm{\psi}_{\B^*_{H_0}} = \norm{\psi}_{\B^*} + \norm{p\psi}_{\B^*}+\norm{H_0\psi}_{\B^*}
\]
into a relatively compact subset of $\B$, where $\B$ is the Besov (or Agmon-Hörmander) space and $\B^*$ its dual. A similar condition is derived in \cite{Amrein_C0} using abstract Mourre machinery. As shown by \cite[Theorem 14.4.2]{Hormander2} this condition is satisfied if $V^{sr}$ is dominated by a radial \emph{decreasing} $L^1$ function.  The compactness of $V^{sr}:\B^*_{H_0}\to \B$ essentially reduces the spectral properties of $H_0 + V^{sr}$ to that of $H_0$ using perturbative arguments. In the setting of the present paper $W_0$ might not map $\B^*_{H_0}$ into $\B$, let alone compactly. Hörmander only considers $O(r^{-1-\eps})$ decay in the follow-up \cite{Hormander4} on long-range scattering. We finally mention the class of Enss potentials (cf.\ \cite{Enss}) as described in \cite[Section XI.17]{Reed3} which naturally incorporates a decreasing $L^1$ condition as well.

\subsection{Basic setting and notation}
We denote by $\mathcal{L}(X,Y)$ the space of bounded operators $X\to Y$ and write $\mathcal{L}(X)=\mathcal{L}(X,X)$.

Recall the Besov spaces (also known as Agmon-Hörmander spaces cf.\ \cite{AgmonHormander}) mentioned above: Set $F_k = 1_{\{2^{k-1}\leq r < 2^k\}}$ for $k\geq 1$. Then define
\begin{align*}
	\B &= \{\psi \in L^2_{loc} \mid \norm{\psi}_\B < \infty\}, \quad \norm{\psi}_\B = \sum_{k = 1}^\infty 2^{k/2} \norm{F_k \psi}\\
	\B^* &= \{\psi \in L^2_{loc} \mid \norm{\psi}_{\B^*}< \infty\}, \quad \norm{\psi}_{\B^*} = \sup_{k\geq 1} 2^{-k/2}\norm{F_k \psi}\\
	\B^*_0 &= \{\psi \in L^2_{loc}\mid \lim_{k\to \infty} 2^{-k/2}\norm{F_k \psi} = 0\}.
\end{align*}
We equip the spaces with their natural norms and identify the dual space of $\B$ with $\B^*$ through the $L^2$-pairing. It is well known that $\B$ and $\B^*$ are non-reflexive Banach spaces with $C^\infty_c$ dense in $\B$. The space $\B_0^*$ is exactly the closure of $C^\infty_c$ in $\B^*$, hence the notation. It is easily verified that $\B$ is separable. This in particular means that any norm bounded subset of $\B^*$ is metrizable and relatively compact in the weak-star topology, which is useful when proving the LAP. See \cite[Theorem 3.16]{Rudin_funct}.

Denote by $L^2_s(\R^d) = r^{-s}L^2(\R^d)$, $s\in \R$, the weighted space of all $\psi \in L^2_{loc}$ such that $r^s \psi \in L^2$ with the natural norm. The Besov spaces are interpolation spaces in the weighted $L^2$ scale with the following continuous inclusions:
\[
L^2_s \subseteq \B \subseteq L^2_{1/2} \subseteq L^2 \subseteq L^2_{-1/2} \subseteq \B^*_0 \subseteq \B^* \subseteq L^2_{-s}, \quad s>1/2.
\]
We denote by $H^2(\R^d) = W^{2,2}(\R^d)$ the standard order two $L^2$ Sobolev space. Other Sobolev spaces are denoted similarly. 

Like in the abstract Mourre theory, our commutator estimates are based on a well chosen 'conjugate operator'. We shall consider the symmetric radial differential operator
\begin{equation}\label{B,exp}
A = \Im(\omega\cdot\nabla) =  \tfrac{1}{2}(p^*\omega + \omega^*p) = \tfrac{1}{2}(p\cdot \omega + \omega\cdot p) = p\cdot \omega +\tfrac{i}{2}\Delta r = \omega\cdot p -\tfrac{i}{2} \Delta r
\end{equation}
with domain $D(A) = H^1$. Note that the vector field $\omega$ is smooth and bounded, so $A$ is relatively $H_0$-bounded, which is in contrast to the usual generator of radial dilations. The operator $A$ is also used in \cite{ItoSkibsted20} and \cite{ItoSkibsted21}, where a self-adjoint realization is also provided. Symmetry on $H^1$ suffices in our setting. Note the basic commutator identity
\[
i[A,f] = i(Af-fA) = \omega \cdot \nabla f, \quad f\in H^1_{loc}
\]
as quadratic forms on $C^\infty_c$. Using the operator $A$ we can decompose $H = H_0 + V$ into radial and spherical components in the following way. Consider the Laplace-Beltrami type operator associated to $A$:
\[
L = p^* M p \quad M = 1- \omega\omega^* = r \nabla^2 r,
\]
where $\nabla^2 r$ is the Hessian of $r$. We can then write
\begin{equation}\label{B^2}
	H = H_0 + V = L + A^2 + V + E_1,
\end{equation}
where
\begin{equation}\label{E_1}
	E_1 = \tfrac{1}{4} (\Delta r)^2 + \tfrac{1}{2}\omega\cdot \nabla \Delta r = \tfrac{1}{4}\big((d-1)(d-3)r^{-2} +(4d-10)r^{-4} + 7r^{-6}\big).
\end{equation}
Essentially the same decomposition is used in \cite{ItoSkibsted20}, although they absorb $E_1$ and $V$ into the 'effective potential' $q = E_1+V$, and they also divide out (for $r$ large) by $\abs{\omega}^2 = 1-r^{-2}$ in the matrix $\omega\omega^*$ in the definition of $L$. The explicit nature of $r$ in the present paper renders this unnecessary.

Finally introduce a function $\chi \in C^\infty(\R)$ such that
\[
\chi(t) = \begin{cases}
1,&  t\leq 1\\
0,&  t\geq 2
\end{cases}
, \quad \chi' \leq 0,
\]
and fix the smooth cut-offs $\chi_n = \chi_n(r) = \chi(r/2^n)$ and $\chi_{m,n} = \chi_n (1-\chi_m)$, $n,m\geq 0$. Here and throughout we use primes to denote derivatives of radial functions, i.e.\ if $f = f(r)\in C^1([1,\infty))$, for instance, then $\nabla f = \omega f'$.
\subsection{Results}
\subsubsection{Rellich type theorem}
We start with a uniqueness criterion to the homogeneous Helmholtz equation 
\begin{equation}\label{helm}
(H-\lambda)\phi = 0, \quad \lambda>0.
\end{equation}
If $\phi \in H^2_{loc}$ and $\phi$ satisfies \eqref{helm} in the distributional sense, we call $\phi$ a generalized eigenfunction with eigenvalue $\lambda$. If in addition $\phi \in L^2$, then $\phi$ is naturally a genuine eigenfunction. 
\begin{thm}\thlabel{Rellich}
Suppose \thref{cond1}. If $\phi \in H^2_{loc}\cap \B^*_0$ is a generalized eigenfunction with positive eigenvalue, then $\phi = 0$. In particular, $H$ has no positive point spectrum.
\end{thm}
Our proof is actually slightly more general. Only assuming $W_0 = o(r^{-1})$ in \thref{cond1} we are able to show that any generalized eigenfunction $\phi \in H^2_{loc}$ with positive eigenvalue $\lambda>0$ such that $\phi \exp(C_\lambda \smallint_1^r W_0 \, ds) \in \B^*_0$ must vanish identically. Here $C_\lambda$ is an explicitly computable constant only dependent on $\lambda$. See \thref{Rellich2} for an exact statement. In particular, our proof establishes the absence of positive point spectrum under classical $o(r^{-1})$ decay, which has been extensively studied in the literature. We mention \cite{Agmon70, Kato59, HerbstFroese, Simon69}. Comparatively our setting is very general. Taking $W_0 \in L^1(dr)$ obviously yields \thref{Rellich}. The sharp $\B^*_0$ condition will play a crucial role in the proofs of all theorems below. 

\subsubsection{LAP bounds}
For a set $I\subseteq (0,\infty)$, let
\[
I_\pm = \{\lambda \pm i \eps \mid \lambda \in I, \, 0<\eps<1\}.
\]
This should be read as two separate sets $I_+$ and $I_-$. The upper bound on $\eps$ is technically redundant in all theorems stated below. In the following, and elsewhere in the paper, we employ the notation $\sch{T}$ for the quadratic form associated to a symmetric operator $T$. 

\begin{thm}\thlabel{unif_bound}
Suppose \thref{cond1} and let $I\subseteq (0,\infty)$ be a compact set. There exist $C>0$ so that uniformly for all $\psi \in \B$ and $z\in I_\pm$
\[
\norm{R(z)\psi}^2_{\B^*} + \norm{AR(z)\psi}^2_{\B^*} + \sch{p^*\nabla^2 r p}_{R(z)\psi} \leq C \norm{\psi}^2_\B.
\]
\end{thm}
Since $A$ and $p^*\nabla^2 r p$ essentially act orthogonally, \thref{unif_bound} can be extended to also include $pR(z)\psi$. Indeed, simply note
\[
2^{-k}\norm{F_k pR(z)\psi}^2 \leq \sch{p^*F_k\nabla^2 rp}_{R(z)\psi} + 2^{-k}\norm{F_k \omega^*pR(z)\psi}^2
\]
and take $\sup$ over $k\in \N$. As such, \thref{unif_bound} in particular states that $R(z)$, $z\in I_\pm$, is uniformly bounded in $\mathcal{L}(\B,\B^*_{H_0})$, where $\B^*_{H_0}$ is the space mentioned in the introduction. 

\thref{unif_bound} is complemented by the following list of smoothness bounds. These are stated in terms of radial functions $W=W(r)$ which behave like $W_0$, that is
\begin{equation}\label{W_general}
W \geq 0, \quad W = o(r^{-1}), \quad W \in L^1([1,\infty),dr).
\end{equation}
Throughout $W^{1/2}$ will play the role usually taken by $r^{-1/2-\eps}$ in the literature.
\begin{prop}\thlabel{W_resolvent}
Suppose \thref{cond1}. Fix functions $W_1$, $W_2$ that satisfy \eqref{W_general} and consider a compact set $I\subseteq (0,\infty)$. The following are bounded uniformly for $z\in I_\pm$:
\begin{enumerate}
	\item The quadratic form $\psi \to \sch{p^*\nabla^2 r p}_{R(z)W_2^{1/2}\psi}$ on $L^2$,
	\item $W_1^{1/2}R(z)$ and $W_1^{1/2}pR(z)$ in $\mathcal{L}(\B,L^2)$,
	\item $R(z)W^{1/2}_1$ and $pR(z)W_1^{1/2}$ in $\mathcal{L}(L^2,\B^*)$,
	\item $W_1^{1/2}R(z)W_2^{1/2}$ and $W_1^{1/2}pR(z)W_2^{1/2}$ in $\mathcal{L}(L^2)$.
\end{enumerate}
\end{prop}
We shall use all of the the bounds from \thref{unif_bound,W_resolvent} in the proofs of \thref{unif_rad,LAP2} stated below. Our smoothness bounds should be compared to Yafaev's result \cite[Theorem 11.1.1]{Yafaev_analytic}. Standard applications of LAP bounds include absence of singular continuous spectrum and asymptotic completeness in the short-range case. See \cite[Chapter 4]{Yafaev_general}, for instance.

\subsubsection{Limiting absorption principle and radiation conditions}
Radiation condition bounds hinge on the observation that $(\partial_r \mp i\sqrt{z})R(z)\psi$ inherits in the limit $\pm\Im z\to 0^+$ 'excess' decay from $\psi$ and $W_0$. Here 'excess' should be understood as more than necessary decay in their respective admissible classes. For instance, consider the case $W_0 = O(r^{-1-2\eps})$ and set $h(r)^2 = r^{2\eps'}$ with $0<\eps'<\min\{1,\eps\}$ so that $h^2 W_0$ remains classically short range (the requirement $\eps<1$ is a technical artefact). The 'usual' radiation condition bound then asserts that if $\psi \in h^{-1}\B$, then $(\partial_r \mp i\sqrt{z})R(z)\psi$ is uniformly bounded in $h^{-1}\B^*$ for $z \in I_\pm$ for all compact $I\subseteq (0,\infty)$, i.e.\ that $(\partial_r \mp i\sqrt{z})R(z)\psi$ inherits the additional decay factor of $h = r^{\eps'}$ from $\psi$ and $W_0^{1/2}$ all the way to the 'cut' $\pm \Im z >0$. This is in contrast to $R(z)\psi$ where the space $\B^*$ is optimal even for $\psi \in C^\infty_c$ and $V = 0$. To establish radiation condition bounds in our minimal decay setting $W_0 \in L^1([1,\infty),dr)\cap o(r^{-1})$, in place of powers of $r$ we quantify the excess decay in terms of more general (possibly asymptotically smaller) functions $h = h(r) \in C^1$ that satisfy
\begin{equation}\label{h req}
h>0, \quad 0\leq h' \leq \beta_0 r^{-1} h, \quad h^2 W_0 = o(r^{-1}), \quad h^2 W_0 \in L^1([1,\infty),dr),
\end{equation}
where $\beta_0 <1$. It is elementary to verify that there always exists a function $h$ satisfying \eqref{h req} so that, in addition,  $h(r)\to \infty$ as $r\to \infty$ since there are some 'wiggle room' in the $L^1$ and $o(r^{-1})$ conditions. One could take $h = (\log r)^{\eps'}$ (for $r$ large) with $0<\eps'<\eps$ when $W_0 = O(r^{-1}(\log r)^{-1-2\eps})$, for instance. Note also that \eqref{h req} automatically requires $h\leq C r^{\beta_0} = o(r^{1})$ by Grönwall's inequality, or more precisely
\[
(s/t)^{\beta_0} \leq h(s)/h(t) \leq (t/s)^{\beta_0}, \quad 1\leq s\leq t.
\]
It is of separate interest in the development of stationary scattering theory to relax the requirement $\beta_0 <1$ in \eqref{h req} to archive so-called 'strong radiation condition bounds'. See \cite{HerbstSkibsted_91, ItoSkibsted_strong} for results in this direction.

Essentially following \cite{ItoSkibsted20}, in place of $\sqrt{z}$ we consider a phase $a = a_z$ constructed as follows: For each $\lambda \in (0,\infty)$, pick $r_\lambda \geq 1$ such that
\[
\lambda - V^{lr}-E_1 > \lambda/2  \text { when } r\geq \tfrac{1}{2} r_\lambda,
\]
which is possible since $V^{lr} + E_1 = o(1)$. Here $E_1$ is given by \eqref{E_1}. We may furthermore assume $\lambda \to r_\lambda$ is decreasing. Then define
\[
\eta_\lambda(r) = 1 - \chi(2r/r_\lambda)
\] 
which is supported in $r\geq \tfrac{1}{2}r_\lambda$ and constantly equal to one for $r\geq r_\lambda$. For $z = \lambda \pm i \eps$ with $\eps\geq 0$ and $\lambda>0$, set 
\begin{equation}\label{a_def}
a = a_z = \pm \eta_\lambda \sqrt{z-V^{lr}-E_1}.
\end{equation}
We choose the positive branch of the square root. The sign $\pm$ in front is solely to ensure $\Im(a)\geq 0$, while the cut-off $\eta$ allows us to take derivatives of $a$. Again, $a = a_{\lambda\pm i \eps}$ should be read as two separate identities (also when $\eps=0$). The decomposition \eqref{B^2} can now be refined to
\begin{equation}\label{A+a}
H-z = (A+a)(A-a) + L + V^{sr} - E_2
\end{equation}
with error
\begin{equation}\label{E_2}
E_2 = (1-\eta_\lambda^2)(z-V^{lr}-E_1) + i\omega\cdot \nabla a.
\end{equation}
The phase $a_z$ is chosen in particular to minimize the commutation errors arising from $E_2$. We can now state our radiation condition bounds:

\begin{thm}\thlabel{unif_rad}
Suppose \thref{cond2} and let $I\subseteq (0,\infty)$ be compact. Consider a function $h$ that satisfies \eqref{h req}. There exists $C>0$ so that for any $\psi \in h^{-1}\B$ and $z\in I_\pm$
\[
\norm{h(A-a)R(z)\psi}^2_{\B^*} + \sch{p^* h^2 \nabla^{2} r p}_{R(z)\psi} \leq C \norm{h \psi}_\B^2.
\]
\end{thm}
We complement \thref{unif_rad} with the following radiation condition 'smoothness bounds':
\begin{prop}\thlabel{rad_smooth}
Suppose \thref{cond2}. Consider functions $W_1$, $W_2$ that satisfy \eqref{W_general} and a function $h$ that satisfies \eqref{h req}. Assume moreover that $h^2W_2$ also satisfies \eqref{W_general}. Then, for any compact set $I\subseteq (0,\infty)$,
\begin{enumerate}
	\item $W_1^{1/2}h (A-a) R(z) W_2^{1/2}$ is uniformly bounded for $z\in I_\pm$ in $\mathcal{L}(L^2)$,
	\item $h(A-a)R(z)W_2^{1/2}$ is uniformly bounded for $z\in I_\pm$ in $\mathcal{L}(L^2,\B^*)$.
\end{enumerate}
\end{prop}
As an application we archive the full limiting absorption principle. We remind the reader that the function '$h$' below can be chosen so that $h(r)\to \infty$ as $r\to \infty$.
\begin{cor}\thlabel{LAP2}
	Suppose \thref{cond2}. Let $I\subseteq (0,\infty)$ be a compact set, and fix functions $W$ and $h$ such that $h$ satisfies \eqref{h req} and both $W$ and $h^2W$ satisfy \eqref{W_general}. Then there exists $C>0$ such that
	\begin{equation}\label{unif cont}
	\norm{W^{1/2}\big(R(z)-R(z')\big)W^{1/2}}_{\mathcal{L}(L^2)} \leq  \frac{C}{h(\abs{z-z'}^{-1})}
	\end{equation}
	for all $z, z' \in I_\pm$ with $\abs{z-z'}\leq 1$.
	
	We conclude the LAP: Equip $\B^*$ with the weak-star topology. For each fixed $\psi \in \B$ and $\lambda\in (0,\infty)$, the boundary values
	\[
	R(\lambda\pm i0)\psi = \wslim_{\eps\to0^+} R(\lambda\pm i \eps)\psi \in \B^*
	\]
	exist, and the corresponding extensions
	\[
	\{\lambda\pm i \eps \mid \lambda>0, \, \eps \geq 0\} \ni z \to R(z)\psi \in \B^*
	\]
	are locally uniformly continuous.
\end{cor}
We highlight the explicit nature of the bound \eqref{unif cont}. In the classical case $W_0 = O(r^{-1-2\eps})$, taking $W = r^{-2s}$ with $s>1/2$ and $h = r^{\eps'}$ with $0< \eps' <\min\{s-1/2, \eps, 1\}$ shows that the resolvent $z \to R(z)\in \mathcal{L}(L^2_s, L^2_{-s})$ is Locally Hölder continuous with parameter $\eps'$. This is in agreement with \cite[Corollary 1.11]{ItoSkibsted20}. Similarly, if $W_0 = O(r^{-1}(\log r)^{-1-2\eps})$ we obtain the new bound
\[
\norm{R(z)-R(z')}_{\mathcal{L}(L^2_s,L^2_{-s})} \leq C \Big( \log\big(\abs{z-z'}^{-1}\big)\Big)^{-\eps'}
\]
for any $s>1/2$, $0<\eps'<\eps$, and $\abs{z-z'}$ sufficiently small.

By uniform boundedness it is naturally also possible to take the weak-star limit $\eps\to 0^+$ for $pR(\lambda\pm i \eps)\psi$ and $H_0R(\lambda\pm i \eps)\psi$ in $\B^*$. In particular, the limiting resolvents $R(\lambda\pm i0)$ map into $H^2_{loc}$, or more precisely the space $\B^*_{H_0}$ considered in the introduction (and in \cite{Hormander2}). Weak-star continuity therefore allows us to extend the conclusions of \thref{unif_bound,W_resolvent,unif_rad,rad_smooth} to include the case $z= \lambda\pm i0$ as well. Making use of this we can establish the following Sommerfeld uniqueness result which generalizes \thref{Rellich} to the inhomogeneous case.

\begin{cor}\thlabel{Sommerfeld2}
Suppose \thref{cond2} and fix some $\lambda\in (0,\infty)$. The limiting resolvents $R(\lambda\pm i 0)\in \mathcal{L}(\B,\B^*)$ are uniquely characterized by the following: For any function $h$ that satisfies \eqref{h req} and any $\psi \in h^{-1}\B$,
\begin{enumerate}
	\item $R(\lambda \pm i 0)\psi \in H^2_{loc}\cap h\B^*$ and $(H-\lambda)R(\lambda\pm i0)\psi = \psi$ in the distributional sense,
	\item $(A-a_{\lambda\pm i 0})R(\lambda\pm i 0)\psi \in h^{-1}\B^*_0$.
\end{enumerate}
\end{cor}

\section{Rellich's theorem}
In this section we first quickly describe the scheme of 'commutators with weight inside' developed by Ito and Skibsted in \cite{ItoSkibsted20}. We then prove \thref{Rellich} in a manner similar to that of \cite[Theorem 1.4]{ItoSkibsted20}.

\subsection{Commutators with weight inside}
Following the method developed in \cite{ItoSkibsted20} we shall consider 'commutators' of the form
\begin{equation}\label{weight_inside}
\sch{2\Im(AfH)}_\psi = \sch{i(HfA-AfH)}_\psi = i\pair{fA\psi}{H\psi} - i \pair{H\psi}{fA\psi},
\end{equation}
for $\psi \in C^\infty_c$, where the choice of weight $f$ depends on the concrete implementation. In the proofs of \thref{Rellich,unif_bound,W_resolvent} we use a subtle modification of the corresponding weights in \cite{ItoSkibsted20}. Namely, if a weight $f = f(r)$ is used in \cite{ItoSkibsted20}, we consider instead (essentially) $\tilde{f} = f \exp(K\smallint_1^r W_0 \, ds)$ where $K$ is a large constant. The additional term is used to control the commutation errors arising from the potential $V$ in \eqref{weight_inside}. Compared to \cite{ItoSkibsted20}, this is the key technical difference (apart from geometry) that facilitates the relaxation from $O(r^{-1-\eps})$ to $L^1\cap o(r^{-1})$. Here it should be noted that $W_0 = O(r^{-1-\eps})$ errors are easily manageable on the Besov scale, but this is not the case for $W_0 \in L^1(dr)\cap o(r^{-1})$ (\thref{W_resolvent} actually shows that these errors are manageable but this is non-trivial). The additional term, however, severely restricts the regularity of $\tilde{f}$. Given $f = f(r) \in C^\infty$, it is clear that $\tilde{f}$ is locally Lipschitz, or equivalently $\tilde{f}\in W^{1,\infty}_{loc}$, and the chain rule applies:
$\nabla\tilde{f} = \omega\cdot\big(f' \exp(K\smallint_1^r W_0 \, ds) + KW_0\tilde{f}\big)$
weakly and strongly almost everywhere. The low regularity of $\tilde{f}$ necessitates some care when expanding the commutator \eqref{weight_inside}. We finally mention that the 'commutators' employed in \cite{Yafaev_analytic,Vakulenko} by Yafaev and Vakulenko also use conjugate operators which include a term like $\exp(K\smallint_1^r W_0\, ds)$ to control errors arising from $V$. Their methods are fundamentally different, however.

We expand the commutator \eqref{weight_inside} according to the decompositions from \eqref{B^2} and \eqref{A+a}. To this end, the following explicit computation is obviously useful. All computations below require little more that integration by parts and as such are completely routine.

\begin{lemma}\thlabel{DL}
Let $f = f(r)$ be locally Lipschitz as a function of $r$ on $[1,\infty)$. Then, as forms on $C^\infty_c$,
\begin{align*}
	2\Im(AfL) &= - \Im(\Delta r f' r^{-2} \omega^* p) - \tfrac{1}{2}\operatorname{div}(f r^{-2}(\Delta r)' \omega)\\
	&\quad + p^* (2f r^{-1}-\abs{\omega}^2f') M p + 2p^*\omega r^{-2} f' \omega^*p.
\end{align*}
\end{lemma}
\begin{proof}
	Note that $M\omega = r^{-2}\omega$. This will be used repeatedly. First rewrite
	\begin{align*}
		2\Im(AfL) &= 2\Im(Ap^*fMp) - 2 \Im(A[p^*,f]Mp)\\
		& = 2\Im(Ap^*fMp) + 2 \Re(Af' \omega^*Mp).
	\end{align*}
	We compute each term separably. For the second:
	\begin{align*}
		2\Re(Af' \omega^* M p) &= 2\Re((p^*\omega + \tfrac{i}{2}\Delta r)f' r^{-2}\omega^*p)\\
		&= 2p^*\omega r^{-2} f' \omega^*p - \Im(\Delta r f' r^{-2} \omega^* p).
	\end{align*}
	For the first term, rewrite $2\Im(Ap^*fMp) = 2\Im(p^*\omega p^* fMp) + \Re(\Delta r p^* Mf p)$ and then compute
	\begin{align*}
		\Re(\Delta r p^* Mf p) &= p^* f\Delta r M p - \Im((\Delta r)' f \omega^* Mp)\\
		&= p^* f \Delta r M p - \tfrac{1}{2}\operatorname{div}(r^{-2}(\Delta r)' f \omega),
	\end{align*}
	and 
	\begin{align*}
		2\Im(p^*\omega p^* fMp) &= i\sum_{i,k,j} \big(p_ifM_{ik}p_k\omega_j p_j - p_j \omega_j p_k M_{ik}fp_i)\\
		& = i \sum_{i,j,k} \big(p_i f M_{ik}\omega_j p_k p_j - p_j p_k\omega_j M_{ik}f p_i\big)\\
		&\quad + \sum_{i,k,j}\big(p_i f M_{ik}(\nabla^2 r)_{kj}p_j + p_j (\nabla^2 r)_{jk}M_{ki} f p_i\big)\\
		&=i \sum_{i,j,k} \big(p_j p_i f M_{ik}\omega_j p_k - p_j p_k M_{ik}\omega_jf p_i\big)\\
		&\quad - i\sum_{i,j,k} p_i [p_j,fM_{ik}\omega_j]p_k\\
		&\quad + \sum_{i,j} p_i f (M\nabla^2 r)_{ij}p_j + p_j (\nabla^2 r M)_{ji}f p_i.
	\end{align*}
	Exchanging the roles of $i$ and $k$ in the second term of the first sum we see that the entire sum vanishes since $M$ is symmetric. The third sum is equal to $2p^* f M\nabla^2 rp$ since $M$ and $\nabla^2 r$ commutes. Finally, we can explicitly compute the second sum as
	\begin{align*}
		- i\sum_{i,j,k} p_i [p_j,fM_{ik}\omega_j]p_k = - p^* f\Delta r Mp - p^* M f' \abs{\omega}^2 p + 2p^*\omega r^{-3}f\omega^*p. 
	\end{align*}
	Collecting similar terms we conclude
	\begin{align*}
		2\Im(AfL) &= 2p^*\omega r^{-2} f' \omega^* p - \Im(\Delta r f' r^{-2} \omega^* p) - \tfrac{1}{2}\operatorname{div}(f r^{-2}(\Delta r)' \omega)\\
		&\quad - p^* M f' \abs{\omega}^2p + 2p^*\omega r^{-3} f \omega^*p+ 2p^* f M\nabla^2 rp\\
		&= - \Im(\Delta r f' r^{-2} \omega^* p) - \tfrac{1}{2} \operatorname{div}(f r^{-2}(\Delta r)' \omega)\\
		&\quad + p^* (2f r^{-1}-\abs{\omega}^2f') M p + 2 p^* r^{-2} f' \omega\omega^* p. 
	\end{align*}
	The lemma follows. 
\end{proof}
\subsection{Proof of Rellich's theorem}
Our proof of \thref{Rellich} follow the arguments from \cite{ItoSkibsted20} closely, which in turn shares resemblance with \cite{HerbstFroese}, for instance. The argument is classical in spirit. We first show that any (generalized) eigenfunction with positive eigenvalue has super-exponential decay. We then show that any such eigenfunction must be compactly supported and hence vanish identically by the unique continuation property of $H$. Both arguments hinge on very careful commutator estimates of the form \eqref{weight_inside}. As mentioned in the introduction, we are able to prove a generalization of \thref{Rellich}.

\begin{condition}\thlabel{cond3}
	There is a radial function $W = W(r)\geq 0$ and a splitting $V = V^{sr} + V^{lr}$ into real bounded functions such that:
	\begin{enumerate}
		\item $W = o(r^{-1})$ as $\abs{x} \to \infty$.
		\item $V^{lr}$ has first order distributional derivatives in $L^1_{loc}$, and 
		\[
		V^{lr} = o(1) = o(r^0) \text{ as } \abs{x}\to \infty, \quad \omega \cdot \nabla V^{lr} \leq W.
		\]
		\item $V^{sr}$ satisfies the short-range condition
		\[
		\abs{V^{sr}}\leq W.
		\]
	\end{enumerate}
\end{condition}
This is just \thref{cond1} without the $L^1$ assumption. We then show:
\begin{thm}\thlabel{Rellich2}
Suppose \thref{cond3}. If $\phi\in H^2_{loc}$ is a generalized eigenfunction with positive eigenvalue $\lambda>0$ such that 
\[
\phi \exp\big((\tfrac{4}{\lambda}+1)\smallint_1^r W\, ds\big) \in \B^*_0,
\]
then $\phi = 0$.
\end{thm}
We impose \thref{cond3} throughout the remainder of this section. We have in no way tried to optimize the exponent $\tfrac{4}{\lambda}+1$, although it is technically possible with some extra care. Note that the exponential term is $O(r^\eps)$ for any $\eps>0$ since $W = o(r^{-1})$, so \thref{Rellich2} in particular gives the absence of $L^2_{-1/2+\eps}$ eigenfunctions, which seems optimal without the $L^1$ requirement.

The proof of \thref{Rellich} depends on lower bounds of commutators of the form \eqref{weight_inside}. We consider the weights
\[
f = f_{\alpha,\beta,R,m,n}(r) = \chi_{m,n}\exp\Big(2\alpha r + 2\beta R\big(1-(1+r/R)^{-1}\big)\Big)
\]
where $\alpha \geq 0$, $\beta\in [0,1]$, $n>m\geq 0$, and $R\geq 1$. The same weights are used in \cite{ItoSkibsted20} (with $\delta = 1$). Write $\Theta = 2\alpha r + 2\beta R\big(1-(1+r/R)^{-1}\big)$ and introduce $\theta_0 = 1 + r/R$.  Note that $\Theta = 2r\big(\alpha + \beta \theta_0^{-1}\big)$ and 
\[
\Theta' = 2\alpha + 2\beta \theta_0^{-2}, \quad \Theta^{(n)} = 2\beta  (-1)^{n+1}n! \,r^{1-n} \theta_0^{-2}(1-\theta_0^{-1})^{n-1}, \quad n\geq 2.
\]
As mentioned above we are mainly interested in the modification
\[
\tilde{f} = f g, \quad g = g_K(r) = \exp\big(K\smallint_1^r W\, ds\big), \quad K\geq 1.
\]
\begin{lemma}\thlabel{Rellich_lemma}
	Fix $\lambda, \alpha_0 >0$.
	\begin{enumerate}
		\item There exists $\beta, C, c >0$ and $n_0 \geq 0$ so that for any $\alpha \in[0,\alpha_0]$, $n>m\geq n_0$, and $R\geq 2^{n_0}$, taking $K = \tfrac{4}{\lambda}(\alpha + 2) +2$ in the definition of $g$, as forms on $H^2(\R^d)$,
		\begin{align*}
			2\Im(A\tilde{f}(H-\lambda)) &\geq cr^{-1}\theta_0^{-1}\tilde{f} - C(\chi_{n-1,n+1}^2 + 				\chi_{m-1,m+1}^2)r^{-1}ge^\Theta\\
			&\quad - \Re(\gamma(H-\lambda)),
		\end{align*}	
		where $\gamma$ is a function, $\abs{\gamma} \leq C g e^\Theta$ and $\supp \gamma \subseteq \supp \tilde{f}$.
		
		\item Taking $K = \tfrac{4}{\lambda} +2$ and $\beta=0$ in the definition of $\tilde{f}$, we can choose $C,c >0$ and $n_0\geq 0$ so that uniformly for $\alpha \geq \alpha_0$ and $n>m\geq n_0$, as forms on $H^2(\R^d)$,
		\begin{align*}
			2\Im(A\tilde{f}(H-\lambda)) &\geq c\alpha^2r^{-1}\tilde{f} - C\alpha^2(\chi_{n-1,n+1}^2 + \chi_{m-1,m+1}^2)r^{-1}ge^{2\alpha r}\\
			&\quad - \Re(\gamma(H-\lambda)),
		\end{align*}
		where $\gamma$ is a function, $\abs{\gamma} \leq \alpha C g e^\Theta$ and $\supp \gamma \subseteq \supp \tilde{f}$.
	\end{enumerate}
\end{lemma}
\begin{proof}
It suffices by density to establish the lower bounds in a) and b) as forms on $C^\infty_c$. Our first goal is to derive an explicit and useful lower bound on $2\Im(A\tilde{f}(H-\lambda))$ uniform in the parameters $R\geq 1$, $\beta \in [0,1]$, $\alpha\geq 0$, $n>m\geq 0$, and $1\leq K\leq K_0$, where $K_0 = \tfrac{4}{\lambda}(\alpha_0 +2)+2$ only depends on $\lambda$ and $\alpha_0$. All constants $c,C>0$ below are uniform in these parameters. We shall collect errors in the term
\[
Q = (1+\alpha^2)\big(r^{-2}\tilde{f} + \abs{\chi_{m,n}'}ge^\Theta + \abs{\chi_{m,n}''}ge^\Theta\big) + p^*\big(r^{-2}\tilde{f} + \abs{\chi_{m,n}'}ge^\Theta\big)p.
\]
It is important that no factor of $\alpha$ appears in the second term.

Using \thref{DL}, the decomposition \eqref{B^2} along with the Cauchy-Schwarz and Young inequalities, first bound
\begin{equation}\label{R1}
\begin{aligned}
	2\Im(A\tilde{f}(H-\lambda)) &= \lambda \abs{\omega}^2 \tilde{f}' + A\abs{\omega}^2 \tilde{f}' A -\abs{\omega}^2(\tilde{f}E_1)'\\
	&\quad +2\Im(A\tilde{f}V) +2\Im(A\tilde{f}L)\\
	&\geq \Theta'\tilde{f}\abs{\omega}^2(\lambda-V^{lr}) + A\abs{\omega}^2 \Theta' \tilde{f}A\\
	&\quad + (K-1)AW\tilde{f}A + p^*(2r^{-1}\tilde{f}-\abs{\omega}^2\tilde{f}')Mp\\
	&\quad  + 2p^*\omega \Theta'\tilde{f}r^{-2}\omega^*p + \tilde{f}W\big(K(\lambda-V^{lr})-2\big) - CQ.
\end{aligned}
\end{equation}
We shall spend considerable effort to combine the second, fourth, and fifth terms above. Initially, a direct computation using \eqref{B,exp} shows 
\begin{equation}\label{R2}
\begin{aligned}
	\MoveEqLeft A\abs{\omega}^2 \Theta' \tilde{f}A + p^*(2r^{-1}\tilde{f}-\abs{\omega}^2\tilde{f}')Mp + 2p^*\omega \Theta'\tilde{f}(1-\abs{\omega}^2)\omega^*p\\
	&\geq p^*\omega \abs{\omega}^2\Theta' \tilde{f}\omega^*p - \tfrac{1}{2}\Delta r \Theta'^2 \tilde{f}+ p^*(2r^{-1}\tilde{f}-\abs{\omega}^2\tilde{f}')Mp\\
	&\quad + 2p^*\omega \Theta'\tilde{f}(1-\abs{\omega}^2)\omega^*p - CQ\\
	&\geq 2p^*\omega \Theta' \tilde{f}\omega^*p + p^*\tilde{f}(r^{-1} - KW)Mp + p^* \tilde{f}r^{-1}Mp\\
	&\quad - p^* \abs{\omega}^2 \Theta' \tilde{f}p - \tfrac{1}{2}\Delta r \Theta'^2 \tilde{f} - CQ\\
	&\geq  p^*\omega \tilde{f}(2\Theta'- r^{-1}\theta_0^{-1})\omega^*p + p^*\tilde{f}(r^{-1} - KW)Mp\\
	&\quad -p^*\tilde{f}(\abs{\omega}^2\Theta' -r^{-1}\theta_0^{-1})p- \tfrac{1}{2}\Delta r \Theta'^2 \tilde{f} - CQ.
\end{aligned}
\end{equation}
Note we introduced $p^* \tilde{f}r^{-1}Mp = p^* \tilde{f}r^{-1}(1-\theta_0^{-1})M p + p^* \tilde{f}\theta_0^{-1}r^{-1}Mp \geq p^* \tilde{f}\theta_0^{-1}r^{-1}Mp$. We now bound the first and third terms of \eqref{R2} on the right. For the first term, a very tedious computation shows
\begin{align*}
	p^*\omega \tilde{f}(2\Theta'- r^{-1}\theta_0^{-1})\omega^*p &=(A+\tfrac{i}{2}\abs{\omega}^2\Theta') \tilde{f} (2 \Theta'- r^{-1}\theta_0^{-1}) (A- \tfrac{i}{2}\abs{\omega}^2\Theta')\\
	&\quad -\tfrac{1}{4}\tilde{f} (2 \Theta'- r^{-1}\theta_0^{-1})(\abs{\omega}^2\Theta' + \Delta r)^2\\
	&\quad + \tfrac{1}{2}\operatorname{div}\big(\tilde{f} (2 \Theta'- r^{-1}\theta_0^{-1})(\abs{\omega}^2\Theta' + \Delta r)\omega\big)\\
	&\geq (A+\tfrac{i}{2}\abs{\omega}^2\Theta') \tilde{f} (2\Theta'- r^{-1}\theta_0^{-1}) (A- \tfrac{i}{2}\abs{\omega}^2\Theta')\\
	&\quad + \tfrac{1}{2}\tilde{f}\abs{\omega}^4 \Theta'^3 - \tfrac{1}{4}\tilde{f}\Theta'^2 r^{-1}\theta_0^{-1} + \tilde{f}\Delta r\Theta'^2 + 2\tilde{f}\Theta'' \Theta'\\
	&\quad + KW \tilde{f}\Theta'^2 - CQ.
\end{align*}
Now for the third term in \eqref{R2}, using the general identity
\begin{equation}\label{pgp}
p^*gp = \Re(gH_0) + \Im((\nabla g)\cdot p), \quad g\in H^1_{loc}
\end{equation}
as forms on $C^\infty_c$, we see that
\begin{align*}
	-p^*\tilde{f}(\abs{\omega}^2\Theta' -r^{-1}\theta_0^{-1})p &= -\Re(\tilde{f}(\abs{\omega}^2\Theta' -r^{-1}\theta_0^{-1})(H-\lambda))\\
	&\quad - (\lambda - V)\tilde{f}(\abs{\omega}^2\Theta' -r^{-1}\theta_0^{-1})\\
	&\quad - \Im(\big(\tilde{f}(\abs{\omega}^2\Theta' -r^{-1}\theta_0^{-1})\big)'\omega^* p)\\
	&\geq -\Re(\tilde{f}(\abs{\omega}^2\Theta' -r^{-1}\theta_0^{-1})(H-\lambda))\\
	&\quad - (\lambda - V)\tilde{f}(\abs{\omega}^2\Theta' -r^{-1}\theta_0^{-1})-\tfrac{1}{2}\tilde{f}\Delta r \Theta'^2\\
	&\quad -K\Im(W\tilde{f}\abs{\omega}^2 \Theta' \omega^*p) - \tfrac{1}{2}KW\tilde{f}\Theta'^2\\
	&\quad - \tfrac{3}{2}\tilde{f}\Theta' \Theta'' - \tfrac{1}{2}\tilde{f}\Theta'^3\abs{\omega}^4 + \tfrac{1}{2}\tilde{f}\Theta'^2 r^{-1}\theta_0^{-1}.
\end{align*}
Plugging everything into \eqref{R2} we conclude
\begin{equation}\label{R3}
\begin{aligned}
	\MoveEqLeft A\abs{\omega}^2 \Theta' \tilde{f}A + p^*(2r^{-1}\tilde{f}-\abs{\omega}^2\tilde{f}')Mp + 2p^*\omega \Theta'\tilde{f}(1-\abs{\omega}^2)\omega^*p\\
	&\geq (A+\tfrac{i}{2}\abs{\omega}^2\Theta') \tilde{f} (2 \Theta'- r^{-1}\theta_0^{-1}) (A- \tfrac{i}{2}\abs{\omega}^2\Theta')\\
	&\quad -\Re(\tilde{f}(\abs{\omega}^2\Theta' -r^{-1}\theta_0^{-1})(H-\lambda)) + \tfrac{1}{2}KW\tilde{f}\Theta'^2\\
	&\quad - (\lambda - V)\tilde{f}(\abs{\omega}^2\Theta' -r^{-1}\theta_0^{-1}) + p^*\tilde{f}(r^{-1} - KW)Mp\\
	&\quad  -K\Im(W\tilde{f}\abs{\omega}^2 \Theta' \omega^*p)+\tfrac{1}{4}\tilde{f}\Theta'^2 r^{-1}\theta_0^{-1} + \tfrac{1}{2}\tilde{f}\Theta'\Theta''\\
	&\quad  - CQ.
\end{aligned}
\end{equation}
Luckily all bad terms cancel. Before returning to \eqref{R1} we make the simple bounds
\[
-K\Im(W\tilde{f}\abs{\omega}^2\Theta' \omega^*p) \geq -\tfrac{K}{2} \tilde{f}W \Theta'^2 - \tfrac{K}{2}A\tilde{f}WA,
\]
and using \eqref{pgp} along with the Cauchy-Schwarz and Young inequalities
\begin{align*}
	Q &\leq C(1+\alpha^2)(r^{-2}\tilde{f}+(\chi_{n-1,n+1}^2 + \chi_{m-1,m+1}^2)r^{-1}ge^\Theta)\\
	&\quad + 2\Re((\tilde{f}r^{-2}+\abs{\chi_{m,n}'}ge^\Theta)(H-\lambda)).
\end{align*}
Here we bounded $\abs{\chi_{m,n}'} \leq C r^{-1} (\chi_{n-1,n+1}^2 + \chi_{m-1,m+1}^2)$. Plugging \eqref{R3} into \eqref{R1} and using the two bounds above we see
\begin{equation}\label{R4}
\begin{aligned}
	2\Im(A\tilde{f}(H-\lambda)) &\geq \tilde{f} r^{-1}\theta_0^{-1}(\lambda-V^{lr}) + \tilde{f}W\big(K(\lambda-V^{lr})-2-\Theta'\big)\\
	&\quad + (A+\tfrac{i}{2}\abs{\omega}^2\Theta') \tilde{f} (2 \Theta'- r^{-1}\theta_0^{-1}) (A- \tfrac{i}{2}\abs{\omega}^2\Theta')\\
	&\quad +\tfrac{1}{4}\tilde{f}\Theta'^2 r^{-1}\theta_0^{-1} + \tfrac{1}{2}\tilde{f}\Theta'\Theta'' + p^*\tilde{f}(r^{-1} - KW)Mp\\
	&\quad + (\tfrac{K}{2}-1)AW\tilde{f}A - \Re(\gamma(H-\lambda))\\
	&\quad - C(1+\alpha^2)(r^{-2}\tilde{f}+(\chi_{n-1,n+1}^2 + \chi_{m-1,m+1}^2)r^{-1}ge^\Theta),
\end{aligned}
\end{equation}
where $\gamma = 2C(\tilde{f}r^{-2}+\abs{\chi_{m,n}'}ge^\Theta) +\tilde{f}(\abs{\omega}^2\Theta' -r^{-1}\theta_0^{-1})$ with an unimportant computable constant $C>0$. This finishes the first (and main) part of proof.

Now we fix parameters according to a) and b) of the lemma. First for a), take $K = \tfrac{1}{\lambda}(4\alpha + 8)+2$ in the definition of $\tilde{f}$. Then choose $n_0$ sufficiently large, only dependent on $\alpha_0$ and $\lambda$ so that $\lambda- V^{lr}\geq \tfrac{\lambda}{2}$ in $\supp \tilde{f}$ when $n>m\geq n_0$. Taking $n_0$ larger sill we may also assume $r^{-1}- KW \geq 0$ in $\supp\tilde{f}$ since $W=o(r^{-1})$. Then
\begin{align*}
\tilde{f}W\big(K(\lambda-V^{lr})-2-\Theta'\big)+ (\tfrac{K}{2}-1)AW\tilde{f}A + p^*\tilde{f}(r^{-1}-KW)Mp \geq 0
\end{align*}
uniformly for $\beta\in[0,1]$, $0\leq \alpha\leq \alpha_0$, $n>m\geq n_0$, and $R\geq 2^{n_0}$. Taking $n_0$ larger still and $\beta>0$ small we also archive
\[
\tilde{f} r^{-1}\theta_0^{-1}(\lambda-V^{lr})+\tfrac{1}{4}\tilde{f}\Theta'^2 r^{-1}\theta_0^{-1} + \tfrac{1}{2}\tilde{f}\Theta'\Theta'' - C(1+\alpha^2)r^{-2}\tilde{f} \geq c \tilde{f} r^{-1} \theta_0^{-1}
\]
uniformly in $0\leq \alpha\leq \alpha_0$, $n>m\geq n_0$ and $R\geq 2^{n_0}$. Finally $2\Theta' - r^{-1}\theta_0^{-1}\geq 0$ for $n_0$ large, and then 
\[
 (A+\tfrac{i}{2}\abs{\omega}^2\Theta') \tilde{f} (2 \Theta'- r^{-1}\theta_0^{-1}) (A- \tfrac{i}{2}\abs{\omega}^2\Theta')\geq 0.
\]
The conclusion of a) follows after implementing the bounds in \eqref{R4}.

Now for b), take $\beta = 0$ and $K = \tfrac{4}{\lambda}+2$ in the definition of $\tilde{f}$. Then take the limit $R\to \infty$ in \eqref{R4}, or equivalently $\theta_0^{-1} \to 1$. Choosing $n_0$ sufficiently large like in a) we can bound
\begin{align*}
\MoveEqLeft \tilde{f} r^{-1}(\lambda-V^{lr}) + \tilde{f}W\big(K(\lambda-V^{lr})-2\big) +  (\tfrac{K}{2}-1)AW\tilde{f}A\\
&\quad +(A+\tfrac{i}{2}\abs{\omega}^2\Theta') \tilde{f} (2 \Theta'- r^{-1}) (A- \tfrac{i}{2}\abs{\omega}^2\Theta')+  p^*\tilde{f}(r^{-1} - KW)Mp\\
&\geq 0
\end{align*}
uniformly for $n>m\geq n_0$ and $\alpha\geq \alpha_0>0$. It remains to bound
\[
\alpha^2 \tilde{f}r^{-1} - 2\alpha\tilde{f}W-C(1+\alpha^2)\tilde{f}r^{-2} \geq c \alpha^2 \tilde{f}r^{-1}
\]
by taking $n_0$ larger sill to reach the conclusion in b), which finishes the proof.
\end{proof}

\begin{proof}[Proof of \thref{Rellich2}]
Take $\phi$ and $\lambda$ as stated in the theorem and assume that
\[
\phi \exp\big((\tfrac{4}{\lambda}+1)\smallint_1^r W\, ds\big) \in \B^*_0.
\]
Our first goal is to establish super-exponential decay of $\phi$ in the sense that $e^{\alpha r} \phi \in L^2$ for any $\alpha\geq 0$. As such, introduce 
\[
\alpha_0 = \sup\{\alpha\geq 0 \mid \phi\, e^{\alpha r} \exp\Big(\big(\tfrac{2}{\lambda}(\alpha + 2) + 1\big)\smallint_1^r W \, ds\Big) \in \B_0^*\}.
\]
The set we take $\sup$ over is non-empty by assumption. Assume for a contradiction that $\alpha_0<\infty$. Take $\beta$ and $n_0$ in agreement with assertion a) of \thref{Rellich_lemma}. If $\alpha_0 = 0$ choose $\alpha = 0$. Otherwise take $0\leq \alpha <\alpha_0$ so that $\alpha + \beta > \alpha_0$. This also fixes $K = \tfrac{4}{\lambda}(\alpha+2)+2$. Now consider $\tilde{f}$ from \thref{Rellich_lemma} with these parameters fixed and $n>m\geq n_0$ and $R\geq 2^{n_0}$ arbitrary. Applying the bound a) to the state $\chi_{m-2,n+2} \phi \in H^2(\R^d)$ gives
\[
\norm{r^{-1/2}\theta_0^{-1/2} \tilde{f}^{1/2} \phi}^2 \leq C_m \norm{\chi_{m-1,m+1}\phi}^2 + C_R 2^{-n} \norm{\chi_{n-1,n+1} e^{\alpha r} g^{1/2} \phi}^2,
\]
uniformly for $n>m \geq n_0$ and $R\geq 2^{n_0}$, where $C_R$ and $C_m$ only depends on $R$ and $m$, respectively. We used that $r \theta_0^{-1} \leq R$ and that $\chi_{m-2,n+2}\phi$ vanishes on the 'commutators' involving $H-\lambda$ since $\phi$ is a generalized eigenfunction and $\supp \chi_{m-2,n+2}' \cap \supp \tilde{f} = \emptyset$. Note that $e^{\alpha r}g^{1/2} \phi \in \B^*_0$ by assumption, so we can take the limit $n \to \infty$ to see
\[
\norm{r^{-1/2}g^{1/2} \theta_0^{-1/2}(1-\chi_m)^{1/2} e^{\Theta/2} \phi}^2 \leq C_m \norm{\chi_{m-1,m+1}\phi}^2
\]
uniformly for $R\geq 2^{n_0}$ and $m\geq n_0$. We can then take the limit $R \to \infty$ by monotone convergence
\[
\norm{r^{-1/2}g^{1/2}(1-\chi_m)^{1/2} e^{(\alpha + \beta)r}\phi}^2 \leq C_m \norm{\chi_{m-1,m+1}\phi}^2.
\]
We conclude
\begin{align*}
\abs{\phi} \, e^{\kappa r} \exp\big((\tfrac{2}{\lambda}(\kappa + 2)+1)\smallint_1^r W\, ds\big) \leq C r^{-1/2} g^{1/2}e^{(\alpha+\beta)r} \abs{\phi} \in L^2 \subseteq \B^*_0,
\end{align*}
for any $\alpha_0 < \kappa < \alpha + \beta$, since $W = O(r^{-1})$, which reaches the desired contradiction.

We now prove $\phi$ is compactly supported. To this end we use assertion b) of \thref{Rellich_lemma}, so take $\alpha_0 =1$, say, $\beta= 0$, $K$, and $n_0$ accordingly. Consider $\tilde{f}$ with these parameters fixed and $n>m\geq n_0$ arbitrary. Applying the bound in b) to the state $\chi_{m-2,n+2}\phi$ yields, as before, 
\[
\norm{r^{-1/2} \tilde{f}^{1/2}\phi}^2 \leq C \norm{\chi_{m-1,m+1}e^{\alpha r} g^{1/2} \phi}^2 + C\, 2^{-n} \norm{\chi_{n-1,n+1}e^{\alpha r} g^{1/2}\phi}^2.
\]
Using the established result that $g^{1/2} e^{\alpha r} \phi \in \B^*_0$ for any $\alpha \geq 1$, we can take the limit $n\to \infty$ to see
\[
\norm{r^{-1/2}g^{1/2}(1-\chi_m)^{1/2}e^{\alpha(r-2^{m+2})} \phi}^2 \leq C \norm{\chi_{m-1,m+1}g^{1/2}\phi}^2
\]
uniformly for $\alpha \geq 1$ and $m\geq n_0$. Note that the right hand side does not depend on $\alpha$, so if $(1-\chi_{m+2})\phi \neq 0$, then restricting the left hand side to $\supp (1-\chi_{m+2})$ and taking the limit $\alpha \to \infty$ yields a contradiction. We therefore conclude that $\phi$ is compactly supported. The unique continuation property of Schrödinger operators (cf.\ \cite[Theorem XIII.57]{Reed4}) finally shows that $\phi$ must vanish identically, which gives the claim.
\end{proof}

\section{LAP bounds}
In this section we prove \thref{unif_bound,W_resolvent}. Both proofs rely on a basic commutator estimate of the form \eqref{weight_inside} with suitable weights. We impose \thref{cond1} throughout the entire section. 

We shall consider weights $f = f(r) \in C^1([1,\infty))$ satisfying
\begin{equation}\label{LAPB_f}
	0\leq f \leq \norm{f}_\infty<\infty, \quad 0\leq f' \leq \beta r^{-1} f, \quad  0\leq \beta <2,
\end{equation}
and a fixed function $W = W(r)$ that satisfies \eqref{W_general} with $W\geq W_0$. The requirement that $\beta<2$ is inessential although it simplifies the estimate below somewhat. The fact that $f$ is bounded is key, however. With this is place, we consider the modification
\begin{equation*}
\tilde{f} = \tilde{f}_{K,W}(r) = f(r) \exp\Big(K \int_1^r W(s)\, ds\Big), \quad K\geq 1.
\end{equation*}
Note that the $L^1$ requirement on $W$ makes $\tilde{f}$ bounded.

\begin{lemma}\thlabel{LAPB_lemma}
For any compact set $I\subseteq (0,\infty)$, $0\leq \beta<2$, and any function $W$ that satisfies \eqref{W_general} with $W\geq W_0$, we can choose a parameter $K\geq 1$ and constants $C,c>0$ such that the following holds: Let $f$ be a bounded function that satisfies \eqref{LAPB_f}. For all $z\in I_\pm$ there is a complex function $\gamma = \gamma_z$ such that
\begin{equation}\label{LAPB_Im}
	\begin{aligned}
		2\Im(A\tilde{f}(H-z)) &\geq cf' + cfW + cAf'A + cAfWA + cp^*f\nabla^2 rp\\
		&\quad  - C r^{-3}f - \Re(\gamma(H-z))
	\end{aligned}
\end{equation}
as forms on $H^2(\R^d)$. We have the bound $\abs{\gamma}\leq C \norm{f}_\infty$.
\end{lemma}
\begin{proof}
The argument is significantly simper than that of \thref{Rellich_lemma}. It suffices to establish the bound \eqref{LAPB_Im} as forms on $C^\infty_c$ by density. Consider first $\tilde{f}$ with arbitrary $K\geq 1$ to be fixed later. Our choice of $K$ will only depend on $\inf I$. All constants $C,c, C_1, \dots>0$ below are allowed only to depend on $I$, $\beta$, and $W$. 
Introduce the error
	\[
	Q = f r^{-3} + p^* f r^{-3} p.
	\]
Write $z = \lambda \pm i \eps$ for $\lambda \in I$ and $0<\eps<1$. First expand 
\begin{align*}
2\Im(A\tilde{f}(H-z)) = 2\Im(A\tilde{f}(H-\lambda)) \mp 2\eps \Re(A\tilde{f}).
\end{align*}
Using the identity $\eps = \pm \Re(i(H-z))$, the fact that $A:H^1\to L^2$ is bounded, and that $V$ is bounded we see
\begin{align*}
\mp 2 \eps \Re(A\tilde{f}) &\geq -\eps \tilde{f} - \eps A\tilde{f}A \geq -C_1 \eps \norm{f}_\infty \big(1+ A^2\big) \geq - C_2 \eps \norm{f}_\infty(1+H_0)\\
&\geq -C_3 \eps \norm{f}_\infty - \Re(C_2\eps\norm{f}_\infty(H-z))\\
&= -\Re\big(\norm{f}_\infty(C_2\eps\pm i C_3)(H-z)\big).
\end{align*}
Hence it remains to bound $2\Im(A\tilde{f}(H-\lambda))$ in the required form. By the decomposition \eqref{B^2}, \thref{DL}, and the Cauchy-Schwarz and Young inequalities it should be clear that
\begin{align*}
2\Im(A\tilde{f}(H-\lambda)) &\geq f' (\lambda - V^{lr}) + \tilde{f} W\big( K(\lambda-V^{lr}) - 2\big) + Af'A\\
&\quad + (K-1)A\tilde{f}WA + p^* (\tilde{f}r^{-1} - \tilde{f}')Mp - CQ.
\end{align*}
We used that $\tilde{f}' = O(r^{-1})$ to obtain $O(r^{-3})$ errors. Choosing $K$ sufficiently large only dependent on $\inf I$ we can bound
\[
f' (\lambda - V^{lr}) + \tilde{f} W\big( K(\lambda-V^{lr}) - 2\big)+ (K-1)A\tilde{f}WA \geq c f' + c fW + cAfWA - CQ
\]
since $\tfrac{1}{2}\inf I > V^{lr}$ outside a compact set. Similarly
\[
p^* (\tilde{f}r^{-1} - \tilde{f}')Mp \geq p^*( (2-\beta)\tilde{f}r^{-1} - KW\tilde{f})Mp \geq c p^*f r^{-1}Mp - CQ
\]
since $W = o(r^{-1})$. A final application of the Cauchy-Schwarz and Young inequalities shows
\[
Q \leq 2 \Re(fr^{-3}(H-z)) + C f r^{-3}.
\]
Combining the simple bounds above clearly yields the desired conclusion.
\end{proof}
The appearance of $\norm{f}_\infty$ in the bound of $\gamma$ is necessary, and it explains why the resolvent $R(z)\psi$ does not inherit good decay from $\psi$ in the limit $\Im z \to 0$. This should be compared to \thref{rad_bound} below.

We now prove \thref{unif_bound}. The argument is almost identical to that of \cite[Theorem 1.7]{ItoSkibsted20}. We include it for the readers convenience and to keep the present paper self-contained. We shall also reuse a few arguments later.

\begin{proof}[Proof of \thref{unif_bound}]
	Fix a compact set $I\subseteq (0,\infty)$. Consider the family of functions 
	\[
	f = f_k(r) = 1 - \big(1+r/2^k\big)^{-1}, \quad k\geq 0,
	\]
	which clearly satisfy \eqref{LAPB_f} with $\beta = 1$ uniformly for all $k$. Note also that $\norm{f_k}_{\infty} \leq 1$. We can thus implement \thref{LAPB_lemma} with $W = W_0$:
	\begin{equation}\label{B1}
	\begin{aligned}
	\MoveEqLeft \norm{f'^{1/2}\phi}^2 + \norm{f'^{1/2}A\phi}^2 + \sch{p^* f \nabla^2 r p}_\phi\\
	&\leq C \big(\norm{f^{1/2}r^{-3/2}\phi}^2 + \norm{\psi}_{\B}\norm{\phi}_{\B^*} + \norm{\psi}_{\B}\norm{A\phi}_{\B^*}\big)
	\end{aligned}
	\end{equation}
	where $z\in I_\pm$ (for a fixed sign) and $\phi = R(z)\psi$ for $\psi \in \B$ arbitrary, uniformly in $k\geq 0$. All constants are uniform in such $k$, $\psi$ and $z$ from now on.
	
	It suffices to show that $\norm{\phi}_{\B^*} \leq C \norm{\psi}_\B$. Indeed, noting that $f_k'^{1/2} \geq f_k'^{1/2} F_k \geq F_k 2^{-k/2 - 1}$ for all $k\geq 0$, it would then follow from \eqref{B1} that
	\[
	2^{-k-2}\norm{F_k A\phi}^2 + \sch{p^* f_k \nabla^2 r p}_{\phi} \leq C\big(\norm{\psi}_\B^2 + \norm{A\phi}_{\B^*}\norm{\psi}_\B\big),
	\]
	which concludes the theorem after taking sup over all $k\geq 0$ and applying Young's inequality (perhaps treating the two terms on the left separably).
	
	We turn to the proof that $R(z)$ is uniformly bounded in $\mathcal{L}(\B,\B^*)$ for $z\in I_\pm$. Assume this is not the case, that is
	\begin{equation}\label{B2}
	\norm{\phi_n}_{\B^*} = 1 \text{ for all } n\in \N, \quad \lim_{n\to \infty} \norm{\psi_n}_\B = 0,
	\end{equation}
	 where $\phi_n = R(z_n)\psi_n$ for $(\psi_n) \subseteq \B$ and $(z_n)\subseteq I_\pm$ with a fixed sign. Repeating the argument above then shows that
	 \begin{equation}\label{B3}
	 \sup_{n\in\N} \,\norm{A\phi_n}_{\B^*} <\infty.
	 \end{equation}
	 Considering a subsequence we may assume that $z_n \to \lambda \in \overbar{I_\pm}$, and if $\Im \lambda \neq 0$ then clearly
	 \[
	 \norm{\phi_n}_{\B^*} \leq \abs{\Im z_n}^{-1} \norm{\psi_n}_\B \to 0,
	 \]
	which is an instant contradiction. We therefore assume $\lambda \in I$. By considering a further subsequence we may assume $\phi_n \to \phi$ in the weak-star topology of $\B^*$ since $(\phi_n) \subseteq \B^*$ is uniformly bounded. Then also $\phi_n \to \phi$ weakly in $L^2_{-s}$ for $s>1/2$. The latter convergence is actually strong. Indeed, for $s>t>1/2$, a rewrite 
	\[
	r^{-s}\phi_n = r^{-s}R(i)\psi_n - (i-z_n)(r^{-s}R(i)r^{t})(r^{-t}\phi_n)
	\]
	reveals the strong convergence $r^{-s}\phi_n \to r^{-s}\phi$ in $L^2$ since $r^{-s}R(i)r^t$ is compact and $\psi_n \to 0$. It is then obvious that $(p\phi_n)$ and $(H_0 \phi_n)$ are Cauchy in $L^2_{-s}$ for any $s>1/2$, so 
	\[
	\phi_n \to \phi, \quad p\phi_n \to p\phi, \quad H_0 \phi_n \to H_0\phi
	\]
	in $L^2_{-s}$, where the $p\phi$ and $H_0\phi$ are understood in the distributional sense. In particular $\phi \in H^2_{loc}$, and 
	\[
	\pair{(H-\lambda)\phi}{\rho} = \lim_{n\to\infty} \pair{(H-z_n)\phi_n}{\rho} = 0, \quad \rho \in C^\infty_c,
	\]
	by \eqref{B2}, so $\phi$ is a generalized eigenfunction with eigenvalue $\lambda>0$. 
	
	We now show $\phi \in \B^*_0$ to derive $\phi = 0$ using Rellich's theorem which will lead to our desired contradiction. It follows from \eqref{B2}, \eqref{B3}, and \eqref{B1} with $\phi = \phi_n$ and $\psi = \psi_n$ that
	\[
	2^{-k}\norm{F_k \phi_n}^2 \leq C\big(2^{-k}\norm{r^{-1}\phi_n}^2 + \norm{\psi_n}_{\B}\big) \leq C\big(2^{-k} + \norm{\psi_n}_\B\big),
	\]
	uniformly for $n,k\in \N$, where we employed the bound $0\leq f \leq 2^{-k} r$. Recall $\phi_n \to \phi$ in $L^2_{loc}$ and $\psi_n \to 0$ in $\B$, so taking the limit $n\to \infty$ finally reveals
	\[
	\lim_{k\to \infty} 2^{-k} \norm{F_k \phi}^2 = 0,
	\]
	or equivalently that $\phi \in \B^*_0$. We conclude $\phi = 0$ by \thref{Rellich} which reaches a contradiction, for \eqref{B1} and the strong convergence $\phi_n \to 0$ in $L^2_{-3/2}$ shows  
	\[
	1 = \norm{\phi_n}^2_{\B^*} \leq C \big(\norm{r^{-3/2}\phi_n}^2 + \norm{\psi_n}_\B) \to 0,
	\]
	finishing the proof.
\end{proof}

\begin{proof}[Proof of \thref{W_resolvent}]
	The proof consists of repeated and rather trivial applications of \thref{unif_bound}, \thref{LAPB_lemma}, and the Cauchy-Schwarz and Young inequalities, so we skip the details. The key point is that $W$ appears with positive sign in \thref{LAPB_lemma}. For b), simply apply \thref{LAPB_lemma} with $f = 1$ and $W = W_0 + W_1$ to the states $R(z)\psi$, $z\in I_\pm$ and $\psi \in \B$. In a) and d), use the states $R(z)W_2^{1/2}\psi$, $\psi \in L^2$, and $W = W_0 + W_1 + W_2$ instead. The uniform boundedness of $R(z)W_1^{1/2}$ in $\mathcal{L}(L^2,\B^*)$ follows from b) after taking adjoints. For the last statement in c), take $f = f_k = 1-(1+r/2^k)^{-1}$ and $W = W_0$ applied to the states $R(z)W_1^{1/2}\psi$, $\psi \in L^2$, in \thref{LAPB_lemma}. Using d) we see $\norm{f_k'^{1/2}AR(z)W_1^{1/2}\psi}^2 \leq C\norm{\psi}^2$, restricting the integral region and taking sup over $k\in\N$ shows $AR(z)W_1^{1/2}$ is uniformly bounded in $\mathcal{L}(L^2,\B^*$). An application of a) then finishes the proof.
\end{proof}

\section{Radiation condition and LAP}
In this section we prove \thref{unif_rad,rad_smooth} using a commutator estimate. We then establish \thref{LAP2,Sommerfeld2} as consequences of the radiation condition bounds in a manner similar to that of \cite{ItoSkibsted20}. We impose \thref{cond2} throughout.

\subsection{Radiation condition commutator estimate}
Fix a compact set $I\subseteq (0,\infty)$. For $z = \lambda \pm i \eps$, $\lambda\in I$, $\eps>0$, recall the phase $a = a_z$ from \eqref{a_def} and the corresponding decomposition of $H-z$ in \eqref{A+a}. We shall employ the following easily verified properties of $a$ below:
\begin{equation}\label{a_bound}
E_2 = O(r^{-3} + W_0), \quad \abs{\nabla a} = O(r^{-3} + W_0), \quad \Im(a) \geq 0, \quad \abs{a} \leq C
\end{equation}
uniformly for $z \in I_\pm$. 

The setup for our commutator estimate is similar to that of \thref{LAPB_lemma}. Concretely, we shall this time around consider functions $f = f(r) \in C^1([1,\infty))$ such that
\begin{equation}\label{f_2}
f\geq 0, \quad 0\leq f' \leq \beta r^{-1} f, \quad 0\leq \beta \leq 2
\end{equation}
The restriction that $\beta \leq 2$ is important. We can now handle $W_0$ commutation errors because of \thref{W_resolvent}, so we need no modification $\tilde{f}$ of $f$. Note also that $f$ will not be bounded in applications, which is distinctly different from \thref{LAPB_lemma}.

\begin{lemma}\thlabel{rad_bound}
For any compact set $I\subseteq (0,\infty)$ and any $0\leq\beta\leq2$, there exist constants $C,c>0$ such that the following holds: Let $f$ be a function that satisfies \eqref{f_2}. Then for all $z\in I_\pm$, as forms on $H^2_1(\R^d) = r^{-1}H^2(\R^d)$,
\begin{align*}
2\Im((A-a)^*f(H-z)) &\geq c (A-a)^* f' (A-a) + (2-\beta)p^*f\nabla^2 r p\\
& \quad - CfW_0 - Cp^*fW_0p - Cr^{-3} f - C p^* f r^{-3}p.
\end{align*}
\end{lemma}
\begin{proof}
The proof is essentially identical to that of \thref{LAPB_lemma} and vastly simpler than \thref{Rellich_lemma}. We again initially consider all computations as forms on $C^\infty_c$. All constants below will only depend on $I$ and $\beta$. Collect errors in the term
\[
Q = r^{-3}f + p^* r^{-3} f p + W_0 f + p^* W_0 f p.
\]
Expand $2\Im((A-a)^*f(H-z))$ according to the decomposition \eqref{A+a} using \thref{DL}, \eqref{a_bound}, and the Cauchy-Schwarz and Young inequalities
\begin{align*}
2\Im((A-a)^*f(H-z)) &= (A-a)^*\big(\abs{\omega}^2 f' + \Im(a)f\big)(A-a)\\
&\quad + 2\Im((A-a)^*f(V^{sr}+E_2))+ 2\Im((A-a)^*fL)\\
&\geq (A-a)^* f' (A-a) + p^*(2r^{-1}f - f')Mp\\
&\quad + 2\Im(Lfa) - CQ\\
&\geq (A-a)^* f' (A-a) + (2-\beta)p^*f \nabla^2 r p\\
&\quad + 2\Im(Lfa) - CQ.
\end{align*}
Making the trivial bound
\begin{align*}
2\Im(Lfa) &= 2\Im(p^*Mfa p) - 2\Re(p^* M \nabla(fa))\\
&= 2p^*Mf \Im(a) p - 2 \Re(p^*\omega f' r^{-2} a) - 2 \Re(p^*M (\nabla a) f)\\
&\geq - CQ
\end{align*}
establishes the lemma as forms on $C^\infty_c$. Noting that $f = O(r^2)$ by Grönwall's inequality, the bound naturally extends to $H^2_1(\R^d)$ by density, giving the claim.
\end{proof}

\subsection{Proof of radiation condition bounds}

\begin{proof}[Proof of \thref{unif_rad}]
Suppose the function $h$ satisfies \eqref{h req} with the constant $\beta_0 <1$ appearing there. Consider the family of weights
\[
f = f_k(r) = h^2(r)\Theta_k^{\beta_1}(r),\quad \Theta_k =1-\big(1+\tfrac{r}{2^k}\big)^{-1},  \quad k\geq 0,
\]
where $\beta_1 >0 $ is chosen such that $2\beta_0 + \beta_1 <2$. Clearly $f\in C^1([1,\infty))$ with 
\[
0\leq f' = 2h h' \Theta^{\beta_1} + \beta_1 h^2 \Theta^{\beta_1} r^{-1} (1-\Theta) \leq (2\beta_0 + \beta_1) r^{-1} f
\]
for any $k\geq 0$. Note that $h^2(r^{-3}+W_0)$ satisfies \eqref{W_general} by assumption, that is $h^{2}(r^{-3} + W_0)\in L^1(dr)\cap o(r^{-1})$. It thus follows from \thref{rad_bound} and \thref{W_resolvent} that
\begin{align*}
\MoveEqLeft \norm{f'^{1/2}(A-a)R(z)\psi}^2 + \sch{p^*f \nabla^ 2p}_{R(z)\psi}\\
&\leq C\big(\norm{(W_0+r^{-3})^{1/2}f^{1/2}R(z)\psi}^2 + \norm{(W_0+r^{-3})^{1/2}f^{1/2}pR(z)\psi}^2\\
&\quad + \norm{f^{1/2}\psi}_\B \norm{f^{1/2}(A-a)R(z)\psi}_{\B^*}\big)\\
&\leq C\big(\norm{h \psi}_\B^2 + \norm{h\psi}_\B \norm{h(A-a)R(z)\psi}_{\B^*}\big),
\end{align*}
uniformly for $k\geq 0$, $z\in I_\pm$, and $\psi \in C^\infty_c$. Here we employed the trivial bound $\Theta \leq 1$. Make the bound $f_k'^{1/2} \geq c 2^{-k/2} F_k h$ and take sup over all $k\geq 0$ (like in the proof of \thref{unif_bound}) to conclude
\begin{align*}
\MoveEqLeft \norm{h(A-a)R(z)\psi}_{\B^*} + \sch{p^* h^2 \nabla^ 2p}_{R(z)\psi}\\
&\leq C\big(\norm{h \psi}_\B^2 + \norm{h\psi}_\B \norm{h(A-a)R(z)\psi}_{\B^*}\big).
\end{align*}
It then finally follows from Young's inequality that
\[
\norm{h(A-a)R(z)\psi}_{\B^*}^2+\sch{p^*h^2 \nabla^2 r p}_{R(z)\psi} \leq C \norm{h\psi}^2_\B
\]
for some $C>0$ uniformly for $\psi \in C^\infty_c$. Here it is important that $h(A-a)\psi \in \B^*$, but this is obvious since $\psi \in C^\infty_c$. Formally commutating the factor $h$ in $h(A-a)R(z)$ to the right (using the $C^1$ condition on $h$) shows that the bound extends to all $\psi \in h^{-1}\B$, finishing the proof. 
\end{proof}

\begin{proof}[Proof of \thref{rad_smooth}]
The proof is similar in spirit to that of \thref{W_resolvent} (using \thref{rad_bound} instead of \thref{LAPB_lemma}), so we skip some details. The assertion a) follows by considering \thref{rad_bound} with weight $f = h^2 \exp(\alpha\int_1^r h^2W_2 + W_1\, ds)$ where $\alpha>0$ is sufficiently small such that $0\leq f' \leq 2r^{-1}f$, evaluating in $R(z)W_2^{1/2}\psi$ for $\psi \in L^2$, and using the two-sided estimates from \thref{W_resolvent} to bound some terms. The statement b) follows from a) by considering $f = h^2 \Theta^{\beta_1}$ like in the proof of \thref{unif_rad} above, again evaluating the forms in $R(z)W_2^{1/2}\psi$. We consider the proof done.
\end{proof}

\subsection{Proof of the LAP and corollaries}
\begin{proof}[Proof of \thref{LAP2}]
The structure of the proof is similar to that of \cite[Corollary 1.11]{ItoSkibsted20}. We first establish the bound \eqref{unif cont}. Let $z,z'\in I_\pm$ for some fixed sign $\pm$, and take $n\in\N$. Decompose
\begin{equation}\label{eq_lap1}
\begin{aligned}
R(z)-R(z') &= \chi_n R(z)\chi_n - \chi_n R(z')\chi_n + \big(R(z)-\chi_nR(z)\chi_n\big)\\
&\quad - \big(R(z')-\chi_nR(z')\chi_n\big).
\end{aligned}
\end{equation}
We first bound the third and fourth terms above. By \thref{W_resolvent}:
\begin{align*}
\norm{W^{1/2}\big(R(z)-\chi_nR(z)\chi_n\big)W^{1/2}} &\leq \norm{W^{1/2}R(z)W^{1/2}(1-\chi_n)}\\
&\quad + \norm{(1-\chi_n)W^{1/2}R(z)\chi_nW^{1/2}} \\
&\leq C \norm{h^{-1}(1-\chi_n)}_\infty
\end{align*}
for a constant $C>0$ uniform in $n\in \N$ and $z, z' \in I_\pm$. All constants are uniform in these parameters from now on. We can bound the fourth term in \eqref{eq_lap1} similarly, whence
\begin{align*}
\MoveEqLeft \norm{W^{1/2}\big(R(z)-\chi_nR(z)\chi_n\big)W^{1/2}} + \norm{W^{1/2}\big(R(z')-\chi_nR(z')\chi_n\big)W^{1/2}}\\
&\leq C \norm{(1-\chi_n)h^{-1}}_\infty.
\end{align*}

For the first and second term in \eqref{eq_lap1}, rewrite
\begin{equation}\label{eq_lap2}
\begin{aligned}
\chi_n R(z)\chi_n - \chi_n R(z')\chi_n &= \chi_nR(z)\big(\chi_{n+1}(H-z') - (H-z)\chi_{n+1}\big)R(z')\chi_n\\
&= (z-z') \chi_n R(z) \chi_{n+1} R(z')\chi_n\\
&\quad - \chi_n R(z) [H_0,\chi_{n+1}]R(z')\chi_n.
\end{aligned}
\end{equation}
We bound each term separably. For the first, further decompose
\begin{align*}
\chi_{n+1} &= (A-a_{\overbar{z}})^*\chi_{n+1}(a_z + a_{z'})^{-1} - \chi_{n+1} (a_z+a_{z'})^{-1}(A-a_{z'})\\
&\quad-[A,\chi_{n+1}(a_z+a_{z'})^{-1}].
\end{align*}
Note here that $\overbar{a_{\overbar{z}}} = -a_z$ and that $(a_z+a_{z'})^{-1}$ is uniformly bounded in $I_\pm$ with $\abs{\nabla (a_z+a_{z'})^{-1}} \leq C W_0 + C r^{-3}$. It thus follows from \thref{rad_smooth,W_resolvent} that
\[
\norm{W^{1/2}(z-z')\chi_nR(z)\chi_{n+1}R(z')\chi_n W^{1/2}} \leq C\abs{z-z'}\big(\norm{h^{-1}\chi_{n+1}}_{\mathcal{L}(\B^*,\B)} + 1 \big). 
\]
Here we used that $\norm{\chi_{n+1}'}_{\mathcal{L}(\B^*,\B)} \leq C$. The second term in \eqref{eq_lap2} is similar. This time decompose
\begin{align*}
-[H_0,\chi_{n+1}] &= 2i \Re(\chi_{n+1}'A)\\
&= i\big(\chi_{n+1}'(A-a_{z'}) + (A-a_{\overbar{z}})^*\chi_{n+1}'-(a_z-a_{z'})\chi_{n+1}'\big).
\end{align*}
Bounding $\abs{a_z-a_{z'}} \leq C\abs{z-z'}$, it follows after repeating the argument above that
\[
\norm{W^{1/2}\chi_nR(z)[H_0,\chi_{n+1}]R(z')\chi_n W^{1/2}} \leq C\big( \abs{z-z'} + \norm{h^{-1}\chi_{n+1}'}_{\mathcal{L}(\B^*,\B)}\big).
\]
Combining estimates we conclude
\begin{align*}
\norm{W^{1/2}(R(z)-R(z'))W^{1/2}} &\leq C \big( \norm{(1-\chi_n)h^{-1}}_\infty + \norm{\chi_{n+1}'h^{-1}}_{\mathcal{L}(\B^*,\B)}\big)\\
&\quad + C \abs{z-z'} \big(1+\norm{\chi_{n+1}h^{-1}}_{\mathcal{L}(\B^*,\B)}\big).
\end{align*}
We now make the bound more explicit. To this end, note that $h$ is increasing and obeys
\[
(s/t)^{\beta_0} \leq h(s)/h(t)\leq (t/s)^{\beta_0} \quad 1\leq s\leq t,
\]
valid since $0\leq h' \leq \beta_0 h r^{-1}$ with $\beta_0<1$. Thus
\[
\norm{(1-\chi_n)h^{-1}}_\infty \leq \frac{C}{h(2^n)},\quad \norm{\chi_{n+1}'h^{-1}}_{\mathcal{L}(\B^*,\B)} \leq \frac{C}{h(2^n)},
\]
and a Cauchy condensation type argument shows 
\[
\norm{\chi_{n+1}h^{-1}}_{\mathcal{L}(\B^*,\B)} \leq C\big( 1 +  \int_1^{2^n} h^{-1}\, dr\big) \leq C \frac{2^n}{h(2^n)}. 
\]
We therefore conclude:
\[
\norm{W^{1/2}(R(z)-R(z'))W^{1/2}} \leq \frac{C}{h(2^n)} \big(1+2^n\abs{z-z'}\big).
\]
Now suppose $\abs{z-z'}\leq 1$. Take $n\in \N$ such that $2^n \leq \abs{z-z'}^{-1}\leq 2^{n+1}$. Then also $h(2^n) \leq h(\abs{z-z'}^{-1}) \leq C h(2^n)$, whence
\[
\norm{W^{1/2}(R(z)-R(z'))W^{1/2}} \leq  \frac{C}{h(\abs{z-z'}^{-1})},
\]
which finishes the proof of \eqref{unif cont}.

We now prove the LAP. To this end, fix $\psi \in \B$ and $\lambda\in (0,\infty)$, and take an open interval $I$ containing $\lambda$ with closure in $(0,\infty)$. For a fixed sign $\pm$, it suffices to show that the map
\[
\Phi_\pm: I_\pm \ni z \to R(z)\psi \in \B^*
\]
is uniformly continuous where $\B^*$ is equipped with the weak-star topology. Indeed, $\Phi_\pm$ maps into a compact (completely) metrizable space by \thref{unif_bound}, so standard extension arguments applies. By uniform boundedness it therefore suffices to show that
\[
\pair{(R(z)-R(z'))\psi}{\rho} \to 0 \text { as } \abs{z-z'}\to 0
\]
for all $\rho\in C^\infty_c$, where $z$, $z'$ are taken in $I_\pm$. Consider the function
\[
W_\psi = \sum_{k=1}^\infty 2^{-k/2} \max\{\norm{F_k \psi}, 2^{-k}\} F_k,
\]
and note $W_\psi$ is chosen exactly such that $W_\psi$ satisfies \eqref{W_general} and $W_\psi^{-1/2}\psi \in L^2$. Using the proved statement \eqref{unif cont} with this $W_\psi$ and a suitable function $h(r)\to \infty$ as $r\to \infty$, we conclude
\[
\abs{\pair{(R(z)-R(z'))\psi}{\rho}} \leq \norm{W_\psi^{-1/2}\rho} \norm{W_\psi^{-1/2}\psi} \norm{W^{1/2}_\psi(R(z)-R(z'))W_\psi^{1/2}} \to 0
\]
as $\abs{z-z}\to 0$, for any $\rho\in C^\infty_c$, which finishes the proof.
\end{proof}

\begin{proof}[Proof of \thref{Sommerfeld2}]
Fix a function $h$ that satisfies \eqref{h req}, $\lambda\in (0,\infty)$, $\psi \in h^{-1}\B$, and a sign $\pm$. We first verify that $R(\lambda\pm i0)\psi$ obeys a) and b) of the theorem. Using \thref{unif_bound,LAP2}, a) follows easily. For b), take a function $h_0 =h_0(r)$ so that $hh_0$ satisfies \eqref{h req}, $h_0(r)\to \infty$ as $r\to \infty$, and $h_0\psi \in h^{-1}\B$. We remark that this is always possible. Applying \thref{unif_rad} and taking a weak-star limit $\eps\to 0^+$ we see $h(A-a)R(\lambda\pm i 0)\psi\in h_0^{-1}\B^*\subseteq \B^*_0$, showing b).

Conversely, suppose $\phi_\pm$ satisfies a) and b) of the theorem. Consider $\phi'_\pm = \phi_\pm - R(\lambda \pm i0)\psi$, and note $\phi_\pm'$ is a generalized eigenfunction with eigenvalue $\lambda$ by a). We show $\phi_\pm' \in \B^*_0$ to conclude $\phi_\pm= R(\lambda\pm i0)\psi$ by \thref{Rellich}. To this end, note
\[
\Im(\chi_n(H-\lambda)) = \Re((A-a)\chi_n') + \chi_n' \Re(a)
\]
for all $n \in \N$ as forms on $H^2_{loc}$. The left hand side vanishes on $\phi_\pm'$, so
\[
0 \leq \sch{\mp \chi_n' \Re(a)}_{\phi'_\pm} = \sch{\pm \Re((A-a)\chi_n')}_{\phi_\pm'} \leq \norm{h^{-1}\phi_\pm'}_{\B^*} \norm{\chi_n'h (A-a)\phi_\pm'}_\B,
\]
and the right hand side goes to $0$ as $n\to \infty$ since $h(A-a)\phi_\pm' \in \B^*_0$ by assumption. We conclude $\phi'_\pm\in \B^*_0$ and therefore $\phi_\pm = R(\lambda \pm i 0)\psi$ which finishes the proof.
\end{proof}

\begin{acknowledgements} 
Part of this work was done during my time as a master's student at Aarhus University. I would like to express my deepest gratitude to Erik Skibsted for introducing me to this topic, his patience and guidance, and for sharing his invaluable insight into scattering theory and mathematics as a whole. This work was partially supported by the VILLUM Centre of Excellence for the Mathematics of Quantum Theory (QMATH) (grant No. 10059).
\end{acknowledgements}

\bibliographystyle{mscplain}
\bibliography{New_sources}

\end{document}